\documentclass{llncs}

\usepackage{graphicx}
\usepackage{tikz}\usetikzlibrary{shapes}

\usepackage{amsmath}
\usepackage{amssymb}
\usepackage{thmtools,thm-restate}

\newcommand{\sa}{synchronizing automata}

\newcommand{\A}{\mathcal{A}}
\newcommand{\B}{\mathcal{B}}
\newcommand{\C}{\mathcal{C}}
\newcommand{\F}{\mathcal{F}}
\newcommand{\G}{\mathcal{G}}
\renewcommand{\S}{\mathcal{S}}
\renewcommand{\O}{\mathcal{O}}

\begin{document}

\begin{frontmatter}

\title{Synchronizing Random Almost-Group Automata\thanks{This work is supported by the French National Agency through ANR-10-LABX-58, Russian Foundation for Basic Research, grant no.\ 16-01-00795, and the Competitiveness  Enhancement Program of Ural Federal University. A major part of the research was conducted during the scientific collaboration under the Metchnikov program arranged by French Embassy in Russia.}}



\author{Mikhail V. Berlinkov\inst{1}, Cyril Nicaud\inst{2}}

\institute{Institute of Natural Sciences and Mathematics, \\
Ural Federal University, Ekaterinburg, Russia, 620062\\
\email{m.berlinkov@gmail.com}
\and
LIGM, Universit\'{e} Paris-Est and CNRS,
5 bd Descartes, Champs-sur-Marne, 77454 Marne-la-Vall\'{e}ee Cedex 2, France\\ \email{cyril.nicaud@u-pem.fr}
}

\maketitle

\begin{abstract}
In this paper we address the question of synchronizing random automata in the critical settings
of almost-group automata. Group automata are automata where all letters act as permutations on the set of states, and they are not synchronizing (unless they have one state). In almost-group automata, one of the letters acts as a permutation on $n-1$ states, and the others as permutations. We prove that this small change is enough  for automata to become synchronizing with high probability. More precisely, we establish that the probability that a strongly connected almost-group automaton is not synchronizing is $\frac{2^{k-1}-1}{n^{2(k-1)}}(1+o(1))$, for a $k$-letter alphabet.
\end{abstract}


\end{frontmatter}

\section{Introduction}
\label{sec:intro}

A deterministic automaton is called \emph{synchronizing} when there exists a word that brings every state to the same state. If it exists, such a word is called \emph{reset} or \emph{synchronizing}. 

Synchronizing automata serve as natural models of error-resistant systems because a reset word allows to turn a system into a known state, thus reestablishing the control over the system. For instance, prefix code decoders can be represented by automata. If an automaton corresponding to a decoder is synchronizing, then decoding a reset word, after an error appeared in the process, would recover the correct decoding process. 

There has been a lot of research done on synchronizing automata since pioneering works of \v{C}ern\'{y}~\cite{Ce64}. Two questions that attract major interest here are whether an automaton is synchronizing and what is the length of shortest reset words if the answer to the first question is `yes'? These questions are also studied from different perspectives such as algorithmic, general statements etc. and in variety of settings, e.g. for particular classes of automata, random settings, \emph{etc.}
The reader is referred to the survey of Volkov~\cite{Vo08} for a brief introduction to the theory of \sa.

One of the most studied direction of research in this field is the long-standing conjecture of \v{C}ern\'{y}, which states that if an automaton is synchronizing, then it admits a reset word of length at most $(n-1)^2$, where $n$ is the number of states of the automaton. This bound is best possible, as shown by \v{C}ern\'{y}. However, despite many efforts, only cubic upper bounds have been obtained so far~\cite{Pin1983,Szykula2017}.

\bigskip

It is the probabilistic settings that interest us in this article. During the attempts to tackle the conjecture of \v{C}ern\'{y}, lots of experiments have been done, showing that random automata seem to be synchronizing with high probability, and  that their reset words seem to be quite small in expectation. This was proved quite recently in a series of articles:
\begin{itemize}
\item Skvortsov and Zaks~\cite{Zaks10} obtained some results for large alphabets (where the number of letters grows with $n$);
\item Berlinkov~\cite{Berl2013RandomAut} proved that the probability that a random automaton is not synchronizing is in $\O(n^{-k/2})$, where $k$ is the number of letters, for any $k\geq 2$ (this bound is tight for $k=2$);
\item Nicaud~\cite{FastSyn} proved that with high probability a random automaton admits a reset
word of length $\O(n\log^3 n)$, for $k\geq 2$ (but with less precise error terms than in~\cite{Berl2013RandomAut}).
\end{itemize}
All these results hold for the \emph{uniform distribution} on the set of deterministic and complete automata with $n$ states on an alphabet of size $k$, where all automata have the same probability. And it is, indeed, the first probability distribution to study. The reader is refered to the survey~\cite{RandomAutSurvey} for more information about random deterministic automata.

\bigskip

In this article we study a distribution on a restricted set of deterministic automata, the \emph{almost-group automata}, which will be defined later in this introduction. In order to motivate our choice, we first need to outline the main features of the uniform distribution on deterministic automata and how they were used in the proofs of the articles cited above.

In a deterministic and complete automaton, one can consider each letter as a map from the set of states $Q$ to itself, which is called its \emph{action}. The action of a given letter in a uniform random automaton is a uniform random mapping from $Q$ to $Q$. Properties of uniform random mappings have been long studied and most of their typical\footnote{In all the informal statements of this article, \emph{typical} means \emph{with high probability} as the size of the object (cardinality of the set, number of states of the automaton, ...) tends to infinity.} statistics are well known. The \emph{functional graph} proved to be a useful tool to describe a mapping; it is the directed graph of vertex set $Q$, built from a mapping $f:Q\rightarrow Q$ by adding an edge $i\rightarrow j$ whenever $j=f(i)$. Such a graph can be decomposed as a set of cycles of trees. Vertices that are in a cycle consists of elements $q\in Q$ such that $f^\ell(q)=q$ for some positive $\ell$. They are called \emph{cyclic vertices}.

The expected number of cyclic vertices in a uniform random mapping on a set of size $n$ is in $\Theta(\sqrt{n})$. This is used in~\cite{FastSyn} and~\cite{Berl2013RandomAut} to obtain the synchronization of most automata. The intuitive idea is that after reading $a^n$, the set of states already shrinks to a much smaller set, in a uniform random automaton; this gives enough leverage, combined with the action of the other letters, to fully synchronize a typical automaton. 

\medskip

In a nutshell, uniform random automata are made of uniform random mappings, and each uniform random mapping is already likely to  synchronize most of the states, due to their inherent typical properties. At this point, it seems natural to look for "harder" random 
instances with regard to synchronization, and it was a common question asked when the authors presented their works. 

In this article, to prevent easy synchronization from the separate action of the letter, we propose to study what we call \emph{almost-group automata}, where the action of each letter is a permutation, except for one of them which has only one non-cyclic vertex. An example of such an automaton is depicted on Fig:~\ref{fig:intro}.

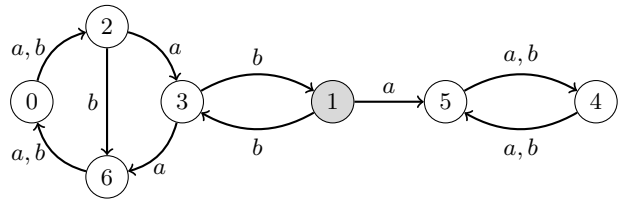
\begin{figure}[h]
\begin{center}
\begin{tikzpicture}
\node[draw,circle] (p0) at (0,0) {0};
\node[draw,circle] (p2) at (1,1) {2};
\node[draw,circle] (p3) at (2,0) {3};
\node[draw,circle] (p6) at (1,-1) {6};

\node[draw,circle,fill=black!15] (p1) at (4,0) {1};
\node[draw,circle] (p5) at (5.5,0) {5};
\node[draw,circle] (p4) at (7.5,0) {4};

\draw[->,thick] (p0) edge[bend left] node[left]{$a,b$} (p2);
\draw[->,thick] (p2) edge[bend left] node[right]{$a$} (p3);
\draw[->,thick] (p3) edge[bend left] node[below]{$a$} (p6);
\draw[->,thick] (p6) edge[bend left] node[left]{$a,b$} (p0);
\draw[->,thick] (p2) edge node[left]{$b$} (p6);

\draw[->,thick] (p3) edge[bend left] node[above]{$b$} (p1);
\draw[->,thick] (p1) edge[bend left] node[below]{$b$} (p3);

\draw[->,thick] (p1) edge node[above]{$a$} (p5);
\draw[->,thick] (p5) edge[bend left] node[above]{$a,b$} (p4);
\draw[->,thick] (p4) edge[bend left] node[below]{$a,b$} (p5);

\end{tikzpicture}
\end{center}
\caption{An almost-group automaton with $7$ states. The action of $b$ is a permutation. The action of $a$ is not, as $1$ has no preimage by $a$; but if state $1$ is removed, $a$ acts as a permutation on the remaining states.\label{fig:intro}}
\end{figure}

Since a group automaton with more than one state cannot be synchronizing, almost-group automata can be seen as the automata with  the maximum number of cyclic states (considering all its letters) that can be synchronizing. The question we investigate in this article is the following.

\smallskip
\noindent\textbf{Question: }For the uniform distribution, what is the probability that a strongly connected almost-group automaton is synchronizing?
\smallskip

For this question, we consider automata with $n$ states on a $k$-letter alphabet, with $k\geq 2$, and try to answer asymptotically as $n$ tends to infinity. We prove that such an automaton is synchronizing with probability that tends to $1$. We also provide a precise asymptotic estimation of the probability that it is not synchronizing. In other words, one can state our result as follows: group automata are always non-synchronizing when there are at least two states, but if one allows just one letter to act not bijectively for just one state, then the automaton is synchronizing with high probability. This suggests that from a probabilistic point of view, it is very difficult to achieve non-synchronization.

This article starts with recalling some basic definitions and notations in Section~\ref{sec:basicdefs}.
Then  some interesting properties of this set of automata regarding synchronization are described in Section~\ref{sec:almostgroup}. Finally, we rely on this properties and some elementary counting techniques to establish our result in Section~\ref{sec:counting}.

\section{Basic Definitions and Notations}
\label{sec:basicdefs}
\noindent\textbf{Automata and synchronization.}
Throughout the article, we consider automata on a fixed $k$-letter alphabet $\Sigma=\{a_0,\ldots,a_{k-1}\}$. Since we are only interested in synchronizing properties, we only focus on
the transition structure of automata: we do not specify initial nor final states, and will never actually consider recognized languages in the sequel. From now on a \emph{deterministic
and complete automaton} (DFA) $\A$ on the alphabet $A$ is just a pair $(Q,\cdot)$, where 
$Q$ is a  non-empty finite set of \emph{states} and $\cdot$, the \emph{transition mapping}, is a mapping from $Q\times A$ to $Q$, where the image of $(q,a)\in Q\times A$ is denoted $q\cdot a$. It is inductively extended to a mapping from $Q\times A^*$ to $Q$ by setting
$q\cdot \varepsilon = q$ and $q\cdot ua=(q\cdot u)\cdot a$, for any word $u\in A^*$ and any letter $a\in A$, where $\varepsilon$ denote the empty word.

Let $\A=(Q,\cdot)$ be a DFA. A word $u\in A^*$ is a \emph{synchronizing word} or a \emph{reset word} if for every $q,q'\in Q$, $q\cdot u=q'\cdot u$. An automaton is \emph{synchronizing} if it admits a synchronizing word. A subset of states $S\subseteq Q$ is \emph{synchronized} by a word $u\in A^*$ if $|S\cdot u|=1$.

Observe that if an automaton contains two or more terminal strongly connected components\footnote{A strongly connected component $S$ is terminal when $S\cdot u\subseteq S$ for every $u\in A^*$.}, then it is not synchronizing. Moreover if it has only one terminal strongly connected component $S$, then it is synchronizing if and only if $S$ is synchronized by some word $u$. For this reason, most works on synchronization focus on strongly connected automata, and this paper is no exception.

\smallskip
\noindent\textbf{Almost-group automata.}
Let $\S_n$ be the set of all permutations of $E_n=\{0,\ldots,n-1\}$. A \emph{cyclic point} of a mapping $f$ is an element $x$ such that $f^\ell(x)=x$ for some positive $\ell$. An \emph{almost-permutation} of $E_n$ is a mapping
from $E_n$ to itself with exactly $n-1$ cyclic points; its unique non-cyclic point is called \emph{dangling point}
(or \emph{dangling state} later on, when we use this notion for automata). Equivalently, an almost-permutation is a mapping that acts as a permutation on a subset of size $n-1$ of $E_n$ and that is not a permutation. Let
$\S'_n$ denote the set of almost-permutations on $E_n$.

An \emph{almost-group automaton} is a DFA such that one letter act as an almost-permutation and all others as permutations. An example of such an automaton is given in Fig.~\ref{fig:intro}. For counting reasons, we need to 
normalize the automata, and define $\G_{n,k}$ as the set of all almost-group automata on the alphabet $\{a_0,\ldots,a_{k-1}\}$ whose state set is $E_n$ and such that $a_0$ is the almost-permutation letter.

\smallskip
\noindent\textbf{Probabilities.} In this article, we equip non-empty finite sets with the uniform distribution, where all elements have same probability. The sets under consideration are often sequences of sets, such as $\S_n$; by abuse of notation, we say that a property \emph{hold with high probability} for $\S_n$ when the probability that it holds, which is defined for every $n$, tends to $1$ as $n$ tends to infinity.

\section{Synchronization of Almost-Group Automata}
\label{sec:almostgroup}

In this section we introduce the main tools that we use to describe the structure of synchronizing and of non-synchronizing almost-group automata.

The notion of a \emph{stable pair}, introduced by
Kari~\cite{KariStable02}, has proved to be fruitful mostly by Trahtman, who managed to use it for solving the famous \emph{Road Coloring Problem}~\cite{TRRCP08}. We make use of this definition in our proof as well, along with some ideas coming from~\cite{TRRCP08}.  

A pair of states $\{p,q\}$ is called \emph{stable}, if for every word $u$ there is a word $v$ such that $p\cdot uv=q\cdot uv$. The \emph{stability} relation given by the set of stable pairs joined with a diagonal set $\{(p,p) \mid p \in Q\}$ is invariant under the actions of the letters and complete whenever $\mathcal{A}$ is synchronizing. The definition on pairs is sound as stability is a symmetric binary relation. It is also transitive whence it is an equivalence relation on $Q$ which is a congruence, i.e. invariant under the actions of the letters. 

Notice also, that an automaton is synchronizing if and only if its stability relation is complete, that is, all pairs are stable.
Because of that, if an automaton is not synchronizing and admits a stable pair, then one can consider a non-trivial factorization of the automaton by the stability relation. So, we aim at characterizing stable pairs in a strongly-connected non-synchronizing almost-permutation automaton, in order to show there is a slim chance for such a factorization to appear when switching to probabilities.  

For this purpose, we need the definition of a \emph{deadlock}, which is a pair that cannot be merged into one state by any word (somehow opposite to the notion of stable pair). A subset $S \subseteq Q$ is called an $F$-clique of $\mathcal{A}$ if it is a set of maximum size such that each pair of states from $S$ is a deadlock. It follows from the definition that all $F$-cliques have  same size and that the image of $F$-clique  by a letter or a word is also an $F$-clique.

Let us reformulate~\cite[Lemma~2]{TRRCP08} for our purposes and present a proof  for self-completeness.
\begin{lemma}
\label{lem:f-clique-diff} If $S$ and $T$ are two  $F$-cliques
such that $S \setminus T = \{p\}$ and $T \setminus S= \{q\}$,
for some states $p$ and $q$, then $\{p,q\}$ is a stable pair.
\end{lemma}
\begin{proof}
By contradiction, suppose there is a word $u$ such
that $\{p\cdot u,q\cdot u\}$ is a deadlock. Then $(S \cup T)\cdot u$ is an
$F$-clique because all its pairs are deadlocks. Since $p\cdot u \neq
q\cdot u$, we have   $|S\cup T| = |S| + 1 > |S|$ contradicting maximality of $S$.\qed
\end{proof}

\begin{lemma}
\label{lem:stable_pair}
Each strongly-connected almost-group automaton $\mathcal{A} \in \mathcal{G}_{n,k}$ with at least two states, admits a stable pair containing the dangling state that is synchronized by $a_0$.
\end{lemma}
\begin{proof}
If $\mathcal{A}$ is synchronizing, then we are done because all pairs are stable. In the opposite case, there must be an $F$-clique $F_1$ of size at least two. 

Let $p_0$ be the dangling state (which is not permuted by $a_0$) and let $d$ be the product of all cycle lengths of $a_0$. Since $\mathcal{A}$ is strongly-connected there is a word $u$ such that $p_0 \in F_1\cdot u$. By the property of $F$-cliques, $F_2 = F_1\cdot u$ and $F_3 = F_2\cdot a_0^{d}$ are $F$-cliques too. Notice that $p_0$ is the only state which does not belong to the cycles of $a_0$ and all the cycle states remains intact under the action $a_0^d$, by construction of $d$. Hence $F_2 \setminus F_3 = \{p_0\}$ and $F_3 \setminus F_2 = \{p_0\cdot a_0^d\}$. Hence,
by Lemma~\ref{lem:f-clique-diff},  $\{p_0,p_0\cdot a_0^d\}$ is a stable pair. This concludes the proof since $p_0\cdot a_0 = p_0\cdot a_0^{d+1}$.\qed
\end{proof}

To characterize elements of $\mathcal{G}_{n,k}$ that are not synchronizing, we build their \emph{factor automata}, which is defined as follows. Let $\A$ be a DFA with stability relation $\rho$. Let $\C=\{C_1$,\ldots, $C_\ell\}$ denote its classes for $\rho$. The \emph{factor automaton}
of $\A$, denoted by $\mathcal{A} / \rho$, is the automaton of set of states $\C$ with transition function defined by $C_i\cdot a = C_j$ in $\mathcal{A} / \rho$ if and only if $C_i\cdot a \subseteq C_j$ in $\A$. Or equivalently, if and only if there exists $q\in C_i$ such that $q\cdot a\in C_j$ in $\A$.

\begin{lemma}
\label{lem:factor_automaton}
If $\A\in\mathcal{G}_{n,k}$ is strongly-connected, then its factor automaton $\mathcal{A} / \rho$ is a strongly-connected permutation automaton.
\end{lemma}
\begin{proof}
Strong-connectivity follows directly from the definition.
If one of the letters was not a permutation on the factor automaton, then there would be a stable class $S$ in $\mathcal{A}$ which has no incoming transition by this letter. It would follow that there is no incoming transition to every state of $S$ in $\mathcal{A}$ either. However, this may happen only for the letter $a_0$ and the (unique) dangling state $p_0$ by this letter. Due to Lemma~\ref{lem:stable_pair}, the dangling state $p_0$ must belong to a stable pair whence there is another state in $S$: this contradicts that $p_0$ is the only state with no incoming transition by $a_0$.\qed
\end{proof}

\begin{lemma}
\label{lem:size_of_components}
Let $\A\in\mathcal{G}_{n,k}$ and let $D$ be the stable class of $\A$ that contains the dangling state $p_0$. Then the set of stable classes can be divided into two disjoint, but possibly empty, subsets $\mathcal{B}$ and $\mathcal{S}$ such that 
\begin{itemize}
\item[$\bullet$] $D \in \mathcal{B}$ and $|B|=|D|$ for every $B \in \mathcal{B}$;
\item[$\bullet$]  $|S|=|D|-1$ for every $S \in \mathcal{S}$;
\item[$\bullet$] The  $a_0$-cycle of  $\mathcal{A} / \rho$ that contains $D$ only contains elements of $\S$ besides $D$;
\item[$\bullet$]  Every other cycle in $\mathcal{A} / \rho$ lies entirely in  either $\mathcal{B}$ or $\mathcal{S}$.
\end{itemize}
\end{lemma}
\begin{proof}
Since stable pairs are mapped to stable pairs, the image of a stable class by any letter must be included in a stable class. 
Recall that by Lemma~\ref{lem:factor_automaton} all letters in $\mathcal{A} / \rho$ act as permutations on the stable classes. Our proof consists in examining the different cycles of the group automaton $\A/\rho$. Let us consider any cycle of a letter $a$ in  $\mathcal{A} / \rho$, made of the stable classes $C_0, C_1, \dots, C_{r-1}$ with $C_j\cdot a \subseteq C_{j+1 \pmod{r}}$, for any $j\in\{0,\ldots r-1\}$.

If $a\neq a_0$ then the letter $a$ acts as a permutation in $\A$, and for each $j$, we have $|C_j| \leq |C_{j+1 \pmod{r}}|$, since $a$ does not merge  pairs of states. Therefore,
\[
|C_0| \leq |C_1| \dots \leq |C_{r-1}| \leq |C_0|.
\]
As a direct consequence, all $|C_j|$  have  same cardinality.

If $a = a_0$, then observe that the same argument can be used when one removes the dangling state $p_0$ and its outgoing transition by $a_0$: the action of $a_0$ on $Q\setminus\{p_0\}$ becomes a well-defined permutation. Henceforth, if this cycle does not degenerate to a simple loop consisting of only $D$, then all the other elements of the cycle are stable classes of size $|D|-1$. And this is the only place where changes of size may happen in $\A/\rho$. The lemma follows from the strong-connectivity of $\mathcal{A} / \rho$. \qed
\end{proof}

Notice that  an almost-group automaton is non-synchronizing if and only if it has at least two stable classes. The following theorem is a consequence of this fact and of Lemma~\ref{lem:size_of_components}.
\begin{theorem}
\label{thm:non-synch-criterion}
A strongly-connected almost-group automaton $\mathcal{A}$ is non-synchro\-nizing if and only if its partitioning described in Lemma~\ref{lem:size_of_components} is such that $|\mathcal{B} \cup \mathcal{S}|>1$.
\end{theorem}

\section{Counting Non-synchronizing Almost-Group automata}
\label{sec:counting}
In this section, we use counting arguments to establish our main result: a precise estimation
of the asymptotic number  of strongly connected almost-group automata that are not synchronizing.

Recall that our working alphabet is $\Sigma=\{a_0,\ldots,a_{k-1}\}$, that $E_n=\{0,\ldots,n-1\}$ and that $\G_{n,k}$ is the set of almost-group automata on $\Sigma$ with set of states $E_n$. Our first counting lemma is immediate.
\begin{lemma}\label{lem:ag automata}
For any $n\geq 1$, there are exactly $(n-1) n!$ almost-permutations of $E_n$. The number of elements of $\G_{n,k}$ is therefore equal to $(n-1)n!^k$.
\end{lemma}
\begin{proof}
An almost-permutation of $E_n$ is characterized by its element with no preimage $x_0$, the way it permutes $E_n\setminus\{x_0\}$ and the image of $x_0$ in $E_n\setminus\{x_0\}$. Since
there are $n$ choices for $x_0$, $(n-1)!$ ways to permute the other elements and $n-1$ choices for the image of $x_0$, the result follows.\qed
\end{proof}

\subsection{Strong-Connectivity}

Our computations below focus on strong-connectivity. We shall need an estimation of the number of strongly connected group automata and almost-group automata. These results are given in 
Lemma~\ref{lem:sc_group_automata} and \ref{lem:sc_automata}. The proofs of these lemmas are kind of folklore, so we moved them into Appendix section~\ref{sec:appendix} to fit into a space limit.

\begin{restatable}{lemma}{primelemma}
\label{lem:sc_group_automata}
There are at most $n(n-1)!^k(1+o(1))$ group automata with set of states $E_n$ that are not strongly-connected. Henceforth, there are $n!^k(1+o(n^{1-k}))$ strongly-connected group automata. 
\end{restatable}

\begin{restatable}{lemma}{secprimelemma}
\label{lem:sc_automata}
The number of not strongly-connected almost-group automata is at most $2(n-1)n(n-1)!^{k}(1+o(1))$. Henceforth, almost-group automata are strongly connected with high probability:
there are $(n-1)n!^k(1+o(n^{1-k}))$ strongly connected elements in $\G_{n,k}$.
\end{restatable}


\subsection{Non-synchronizing Almost-Group Automata: a Lower Bound}
In this section we give a lower bound on the number of strongly connected elements of
$\G_{n,k}$ that are not synchronizing. In order to do so, we build a sufficiently large family
of automata of that kind. The construction of this family is intuitively driven by the structure given in Lemma~\ref{lem:size_of_components} but the formal details of the construction can be done without mentioning this structure.

For $n\geq 3$, let $\F_{n,k}$ be the subset of $\G_{n,k}$, made of the almost-group automata on $\Sigma$ with set of states $E_n$ such that:
\begin{enumerate}
\item there exists a state $p$ that is not the dangling state $p_0$ such that for every letter
$a\neq a_0$, either $p\cdot a = p_0$ and $p_0\cdot a=p$, or $p\cdot a = p$ and 
$p_0\cdot a=p_0$;
\item for at least one letter $a\neq a_0$, we have $p\cdot a = p_0$ and $p_0\cdot a=p$;
\item there exists a state $q\in Q'=E_n\setminus\{p,p_0\}$ such that the action of $a_0$ on $Q\setminus\{p_0\}$ is a permutation with $q$ being the image of $p$;
\item the image of the dangling state by $a_0$ is  $p_0\cdot a_0 = q$.
\item let $q'$ be the preimage of $p$ by $a_0$; if one removes the states $p$ and $p_0$
and set $q'\cdot a_0=q$, then the resulting automaton is a strongly connected group automaton;
\end{enumerate}
The structure of such an automaton is depicted on Fig.~\ref{fig:Fnk}. Clearly from the definition, 
an element of $\F_{n,k}$ is a strongly connected almost group automaton with the dangling state $p_0$.

\begin{figure}[h]
\begin{center}
\begin{tikzpicture}
\node[draw,circle,fill=black!15] (p0) at (-2,1) {$p_0$};
\node[draw,circle] (p) at (-2,-1) {$p$};
\node[draw,circle] (q) at (0,0) {$q$};
\node[draw,circle] (q1) at (-.1,-2.0) {$q'$};
\node (q2) at (1.5,-2) {};
\node (q3) at (2,-1) {};

\draw[->,thick] (p0) -- node[above]{$a_0$} (q);
\draw[->,thick] (p) -- node[below]{$a_0$} (q);
\draw[->,thick,dotted] (q1) edge[bend left] node[left]{$a_0$} (p);
\draw[->,thick,dotted] (q2) edge[bend left] node[below]{$a_0$} (q1);
\draw[->,thick,dotted] (q) edge[bend left] node[below]{$a_0$} (q3);
\draw[->,thick,dotted] (q3) edge[bend left] (q2);

\draw[->] (p) edge[bend left] node[left]{$a_1$}(p0);
\draw[->] (p0) edge[bend left] node[right]{$a_1$}(p);

\draw[->] (p0) edge[loop left] node[left]{$a_2,a_4$}(p0);
\draw[->] (p) edge[loop left] node[left]{$a_2,a_4$}(p);

\node [cloud, draw,cloud puffs=20,cloud puff arc=100, aspect=.9,minimum width=5.2cm,
minimum height=4.8cm] (cloud) at (1.7,-1.1) {};
\node (qp) at (3,.5) {$\mathcal{Q'}$};
\end{tikzpicture}
\end{center}
\caption{The shape of an element of $\F_{n,k}$, with the dangling state $p_0$. \label{fig:Fnk}}
\end{figure}
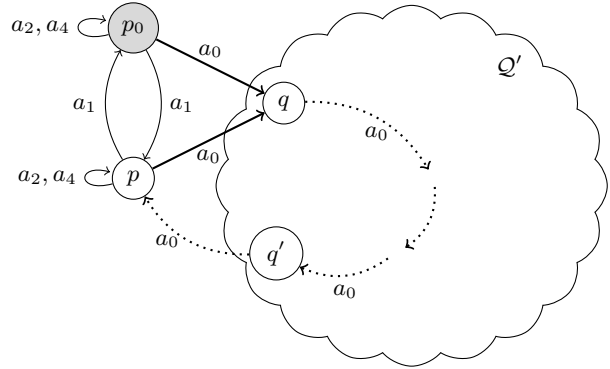

\begin{lemma}\label{lem:Fnk not synchronizing}
For every $n\geq 3$, every automaton of $\F_{n,k}$  is not synchronizing.
\end{lemma}

\begin{proof}
First observe that $\{p_0,p\}$ is the only pair that can be synchronized by reading just a letter, which has to be $a_0$. 
The preimage of $\{p_0,p\}$ is either $\{p_0,p\}$ for $a \neq a_0$ or a singleton $\{q'\}$ otherwise. Hence, no other pair can be mapped to $\{p_0,p\}$ and thus be synchronized by more that one letter.
\qed
\end{proof}

\begin{lemma}
\label{lem:lower_bound}
There are $(2^{k-1}-1)n(n-1)(n-2)(n-2)!^{k}(1+o(n^{1-k}))$ elements in $\F_{n,k}$. Thus there are at least that many strongly connected non-synchronizing almost-group automata.
\end{lemma}
\begin{proof}
From the definition of $\F_{n,k}$, we observe that there are $n(n-1)(n-2)$ ways to choose $p_0$, $p$ and $q$. Once it is done, we choose any strongly connected group automaton $\A'$ with $n-2$ states in $E_N\setminus\{p_0,p\}$; there are $(n-2)!^k(1+o(n^{1-k}))$ ways to do that according to Lemma~\ref{lem:sc_group_automata}. We then change the transition from the preimage $q'$ of $q$ by $a_0$ by setting $q'\cdot a_0 = p$. We set $p\cdot a_0=p_0\cdot a_0 =q$. Finally we choose the actions of the letters $a\in\Sigma\setminus\{a_0\}$ on $\{p_0,p\}$ in
one of the $2^{k-1}-1$ possible ways, as at least one of them is not the identity. This concludes the proof, since all the elements of $\F_{n,k}$ are built exactly once this way.\qed
\end{proof}

Observe that using the definitions of Lemma~\ref{lem:size_of_components}, an element of $\F_{n,k}$ consists of exactly one stable class $\{p_0,p\}$ in $\B$ and $n-2$ stable classes of size $1$ in $\S$.

\subsection{Non-synchronizing Almost-Group Automata: an Upper Bound}
In this section, we upper bound the number of non-synchronizing strongly-connected elements of $\G_{n,k}$  using the characterization of Lemma~\ref{lem:size_of_components}. In the sequel, we freely use the notations used in this lemma (the sets $D$, $\B$, $\S$, \ldots).

Let $b\geq 1$, $s\geq 0$ and $\ell\geq 1$ be three non-negative integers such that $(\ell+1)b + \ell s = n$.
Let $\G_{n,k}(b,s,\ell)$ denote the subset of $\G_{n,k}$ made of the automata such that
$|\B|=b$, $|\S|=s$ and $|D|=\ell+1$.
\begin{lemma}
\label{lem:main_bound}
The number of  non-synchroninzing strongly-connected elements of $\G_{n,k}(b,s,\ell)$ is at most 
\[
\begin{cases}
n! (n-2)!^{k-1}(n-2)(2^{k-1}-1) & \text{if }b=1,\,s=n-2,\text{ and}\ \ell=1,\\
n! \max(1,s) \ell\big(b!s!(\ell+1)!^{b}\ell!^{s}\big)^{k-1}&\text{otherwise}.
\end{cases}
\]
\end{lemma}
\begin{proof}
Our proof consists in counting the number of ways to build, step by step, an element of $\G_{n,k}(b,s,\ell)$.

Firstly, by elementary computations, one can easily verify that the number of 
ways to split $E_n$ into $b$ subsets of size $\ell+1$ and $s$ subsets of size $\ell$ is exactly
\begin{equation}\label{eq:count_partitioning}
\frac{n!}{(\ell+1)!^{b}\ell!^{s}b!s!}.
\end{equation}

Secondly, let us count the number of ways to define the transitions at the level of the factor automaton, i.e. between stable classes, as follows:
\begin{itemize}
\item Choose a permutation on $\mathcal{B}$ in $b!$ ways and on $\mathcal{S}$ in $s!$ ways for each of the $k-1$ letters $a \neq a_0$.
\item Choose which  stable class of $\mathcal{B}$ is the class $D$, i.e. the one containing the dangling state $p_0$, amongst the $b$ possibilities.
\item Choose a permutation for $a_0$ on the $b-1$ classes $\mathcal{B} \setminus\{D\}$ in $(b-1)!$ ways.
\item If $s\neq0$, choose one of the $s!$ permutations of $\mathcal{S}$ for the action of $a_0$ on these classes, then alter the action of $a_0$ the following way: choose the image $D'$ of $D$ by $a_0$ in $\mathcal{S}$ in $s$ ways, then insert it in the $a_0$-cycle: if $D''$ is the former preimage of $D'$, then now $D\cdot a_0 = D'$ and $D''\cdot a_0 = D$ in $\A/\rho$.
\item If $s=0$, then set $D\cdot a_0 = D$ in $\A/\rho$.
\end{itemize}
In total, the number of ways to define the transitions of the factor automaton $\mathcal{A} / \rho$, once the stable classes are chosen is
\begin{equation}\label{eq:factor_count}
(b!s!)^{k-1}b(b-1)!\max(1,s)s! = b!^ks!^{k}\max(1,s).
\end{equation}

Now, we need to define transitions between stable classes for all letters. For all letters but $a_0$, there are $b$ injective transitions between stable classes of size $\ell+1$ and $s$ injective transitions between stable classes of size $\ell$, that is, there are at most $(\ell+1)!^b \ell!^s$ ways to define them for each of the $k-1$ letters. This is an upper bound, as some choices may result in an automaton that is, for instance, not strongly connected. We refine this bound for the
 case $\ell=1, b=1, s=n-2$: one of the letters must swap the states in the single $2$-element class in $\mathcal{B}$ for strong connectivity, so we count just one choice instead of $2$ (for $(\ell+1)!$) to define this letter on this component, that is, we consider only $2^{k-1}-1$ ways to define all permutations on $\mathcal{B}$ in this case, instead of the $((\ell+1)!^b)^{k-1}$ upper bound in the general case (this refinement is used to match our lower bound).

For the action of $a_0$, we additionally choose the dangling state $p_0 \in D$ in $\ell+1$ ways and its image in $D\cdot a_0$ in $\ell$ ways: there are $\ell$ choices in the case where $D\cdot a_0=D$, since $p_0\cdot a_0\neq p_0$, and also when $D\cdot a_0\neq D$, since $D\cdot a_0\in\S$ in this case, according to Lemma~\ref{lem:size_of_components}. Then, it remains to define the injective transitions between the $\mathcal{B} \setminus\{D\}$ blocks in $(\ell+1)!^{b-1}$ ways, and the $s+1$ injective transitions between the $\mathcal{S} \cup \{D'\}$ blocks in $\ell!^{s+1}$ ways, where $D'=D\setminus\{p_0\}$.

Thus, the number of ways to define the transitions between stable classes is at most
$((\ell+1)!^b \ell!^s)^{k-1}\ell(\ell+1)(\ell+1)!^{b-1}\ell!^{s+1}  = \ell(\ell+1)!^{bk}\ell!^{sk}$, in the general case, and  $2(2^{k-1}-1)$ in the case $\ell=1, b=1, s=n-2$.

Putting together \eqref{eq:count_partitioning}, \eqref{eq:factor_count} and this last counting result
yield  the lemma.\qed
\end{proof}

\begin{lemma}
\label{lem:upper_bound}
The number of non-synchroninzing strongly-connected almost-group automata in $\G_{n,k}$ is at most $n(2^{k-1}-1)n!(n-2)!^{k-1}(1+o(1/n))$.
\end{lemma}
\begin{proof}
By Lemma~\ref{lem:lower_bound} and Theorem~\ref{thm:non-synch-criterion}, the number of non-synchroninzing strongly-connected almost-group automata in $\G_{n,k}$  is at most
\begin{equation}
\label{eq:sum_bs}
	n!\sum_{\ell=1}^{\lfloor n/2 \rfloor}\sum_{\{b,s \mid b(\ell+1) + s\ell = n\} } N_{\ell,b,s},
\end{equation}
where $b \geq 1$, $s \geq 0$, and $b+s\geq 2$, and where $N_{\ell,b,s}$ is defined by
\begin{equation}
	N_{\ell,b,s} = 
    \begin{cases}
		\max(1,s) \ell (b!s!(\ell+1)!^{b}\ell!^{s})^{k-1}, & \text{for } (\ell,b,s) \neq (1,1,n-2) \\
        (n-2)!^{k-1}(n-2)(2^{k-1}-1), & \text{for } (\ell,b,s) = (1,1,n-2).
	\end{cases}
\end{equation}
To finish the proof, it will be sufficient to prove that the sum in~(\ref{eq:sum_bs}) is asymptotically equivalent to the term $N_{1,1,n-2}$ since $n!N_{1,1,n-2}$ is asymptotically equivalent to the expression stated in Lemma~\ref{lem:upper_bound}.

To prove this, let us consider the following fraction for $(\ell,b,s) \neq (1,1,n-2)$:
\begin{equation}
\label{leq:fraction1}
\frac{N_{1,1,n-2}}{N_{\ell,b,s}} = \frac{n-2}{\max(1,s) \ell}\frac{(n-2)!^{k-1}(2^{k-1}-1)}{ (b!s!(\ell+1)!^{b}\ell!^{s})^{k-1}} \geq \left(\frac{(n-2)!}{b!s!(\ell+1)!^{b}\ell!^{s}}\right)^{k-1},
\end{equation}
where we used that $n-2 = s\ell + b(\ell+1)-2 \geq s \ell$, as $b$ and $\ell$ are positive; thus $n-2\geq \max(1,s)\ell$ if $s>0$; but it also holds if $s=0$ since $b+s\geq2$.

Observe that, for positive $\ell$ and $m$ we have
\begin{align*}
\frac{(bm)!}{m!^b} &= \left(\frac{1\cdot 2\cdots m}{1\cdot 2\cdots m}\right)
\left(\frac{(m+1)(m+2)\cdots 2m}{1\cdot 2\cdots m}\right)\cdots
\left(\frac{((b-1)m+1)\cdots bm}{1\cdot2\cdots m}\right)\\
& \geq 1^m\cdot 2^m\cdots b^m = b!^m
\end{align*}
Hence, for $m=\ell+1$, we have $\frac{(b(\ell+1))!}{(\ell+1)!^b}\geq b!^{\ell+1}$.
Similarly, one can get that 
\begin{equation}
{\frac{n!}{(b(\ell+1))!}}\frac{1}{\ell!^s} \geq \left(\frac{(b+s)!}{b!}\right)^{\ell}.
\end{equation}
Let $M_{\ell,b,s}=\frac{(n-2)!}{b!s!(\ell+1)!^{b}\ell!^{s}}$, the expression in brackets of (\ref{leq:fraction1}). This quantity can be bounded from below as follows.
\begin{align}
\label{eq:fraction2}
M_{\ell,b,s} &= \frac{1}{n(n-1)b!s!}\frac{(b(\ell+1))!}{(\ell+1)!^{b}} \frac{n!}{(b(\ell+1))!\ell!^s} \\ 
&\geq \frac{b!^{\ell+1}}{n(n-1)b!s!} \left(\frac{(b+s)!}{b!}\right)^{\ell}
\geq \frac{(b+s)!^\ell}{n^2 s!}.
\end{align}

Recall that we want to prove that $M_{\ell,b,s}$ is sufficiently large, so that $N_{1,1,n-2}$ is really
larger than $N_{\ell,b,s}$. Notice that there are at most quadratic in $n$ number of combinations $(\ell, b, s)$ satisfying $b(\ell+1) + s\ell = n$, as for any values $1 \leq b, \ell < n$ there is at most one suitable value of $s$. Therefore, qubic lower bound on $M_{\ell,b,s}$ is enough in general. We distinguish two cases:

\noindent$\triangleright$ If $\ell\geq2$, then $M_{\ell,b,s} \geq {n^{-2}(b+s)!^{\ell-1}}.$ 
If $b+s \geq \ln{n}$, this expression is greater than $\Theta(n^3)$ by Stirling formula. Otherwise, because $b(\ell+1)+s\ell=n$, we have $\ell \geq \frac{n}{\ln{n}}-1$ and as $b+s\geq 2$ the same $\Theta(n^3)$ lower bound holds.

\noindent$\triangleright$ If $\ell=1$, then $s = n-2b$ and $M_{\ell,b,s} \geq  \frac{(n-b)!}{n^2 (n-2b)!}.$
Clearly, this expression decreases as $b$ increases; for $b=3$ it is greater than $\Theta(n)$ (and there is only one such term) and for $b>3$ it is greater than $\Theta(n^3)$. If $b=1$, then $s=n-2$ and this is the term $N_{1,1,n-2}$.
The only remaining case is when $b=2$, $\ell=1$, and $s=n-4$. For this case by (\ref{leq:fraction1}), we get 
\begin{equation}
\frac{N_{1,1,n-2}}{N_{\ell,b,s}} \geq \left(\frac{(n-2)!}{b!s!(\ell+1)!^{b}\ell!^{s}}\right)^{k-1} = \left(\frac{(n-2)!}{8(n-4)!}\right)^{k-1} = \Theta(n^{2(k-1)}).
\end{equation}
Thus, we proved that the sum (\ref{eq:sum_bs}) is indeed asymptotically equal to the term $N_{1,1,n-2}$ multiplied by $n!$.\qed
\end{proof}

\subsection{Main Result and Conclusions}
Now, we are ready to prove our main result on the asymptotic number of strongly connected elements of $\G_{n,k}$ that are not synchronizing.
\begin{theorem}
\label{thm:main}
The probability that a random strongly connected almost-group automaton with $n$ states and $k\geq 2$ letters is not synchronizing is equal to 
\begin{equation}
({2^{k-1}-1}){n^{-2(k-1)}}\left(1+o(1)\right).
\end{equation}
In particular, random strongly connected almost-group automata are synchronizing with high probability as $n$ tends to infinity.
\end{theorem}
\begin{proof}
Lemma~\ref{lem:lower_bound} and Lemma~\ref{lem:upper_bound} give lower and upper bounds on the number of strongly-connected non-synchronizing almost-group automata, which are both equal to $(2^{k-1}-1)n^3(n-2)!^{k}(1+o(1/n))$. We conclude the proof using the estimation on the number of strongly-connected almost-group automata given in
 Lemma~\ref{lem:sc_automata}.\qed
\end{proof}

Thus we obtained a precise asymptotic on the probability for strongly-connected almost group automata of being synchronizable for any alphabet size. As in~\cite{Berl2013RandomAut}, it would be natural to design an algorithm which would verify whether a given random strongly-connected almost group automaton is synchronizing in optimal average time. Another, much more challenging problem, concerns estimation of the expected length of a shortest reset word for random automata in this setting. 

We are thankful to anonymous referees whose comments helped to improve the presentation of the results.



\bibliographystyle{splncs04}
\bibliography{sample.bib}






\newpage
\section*{Appendix}
\label{sec:appendix}

\primelemma*
\begin{proof}
If a group automaton is not strongly-connected, then the set of states can be divided into two non-empty parts $Q_1$ and $Q_2$ such that there is no transition between them (if there is a transition from, say, $Q_1$ to $Q_2$ labeled by $a$, then following this $a$-cycle one finds a transition from $Q_2$ to $Q_1$ at some point). In other words, every non-strongly connected
group automaton can be built by (a) choosing a non-trivial partition of $E_n$ into $Q_1\dot{\cup}
\,Q_2$ such that $|Q_1|\leq |Q_2|$ and (b) choosing a group automaton on the set of states  $Q_1$ and another one  on  $Q_2$. Observe that this construction is not a bijection: if the automaton is made of, for example, three parts, there are several ways to choose $Q_1$ and $Q_2$. Hence, by counting
the number of such decompositions, we are over-counting the number of group automata that 
are not strongly connected, which is fine as we are looking for an upper bound.

The number of group automata that are not strongly connected is therefore at most, using $r$ to denote the cardinality of $Q_1$,
\begin{equation}
Z_{n,k} = \sum_{r=1}^{\lfloor \frac{n}{2} \rfloor}{n \choose r}(r!(n-r)!)^k,
\end{equation} 
since there are $\binom{n}{r}$ ways to choose the elements that are in $Q_1$, and then $r!^k$
group automata on $Q_1$ and $(n-r)!^k$ on $Q_2$. Observe that
\begin{align}
Z_{n,k} &= n! \sum_{r=1}^{\lfloor \frac{n}{2} \rfloor}\left( r!(n-r)!\right)^{k-1} = \\ 
&= n!\left((n-1)!^{k-1} + 2^{k-1}(n-2)!^{k-1} + \sum_{r=3}^{\lfloor \frac{n}{2} \rfloor}\left( r!(n-r)!\right)^{k-1} \right) = \\ 
&= n(n-1)!^k\left(1 + \frac{2^{k-1}}{(n-1)^{k-1}}+n^{k-1}\sum_{r=3}^{\lfloor \frac{n}{2} \rfloor}\binom{n}{r}^{1-k}\right).
\end{align}

This concludes the proof since $\binom{n}{r}^{1-k}\leq \binom{n}{3}^{1-k}$ in the range of the sum and, therefore,
$\sum_{r=3}^{\lfloor \frac{n}{2} \rfloor}\binom{n}{r}^{1-k}\leq \frac{n}2\O(n^{3-3k})$.\qed
\end{proof}

Using the same kind of techniques as in Lemma~\ref{lem:sc_group_automata}, we can obtain an upper bound on the number of almost-group automata
that are not strongly connected.

\secprimelemma*
\begin{proof}
If an automaton of $\G_{n,k}$ is not strongly-connected, then its set of states $E_n$ can be divided into two non-empty parts $Q_1$ and $Q_2$ such that there is no transition from $Q_2$ to $Q_1$. 

To continue the proof, we distinguish two cases, whether there is a transition from $Q_1$ to $Q_2$ or not.
Recall that $p_0$ is the dangling state, with no incoming transition labelled by $a_0$.

\noindent $\triangleright$ If there is a transition $q_1\xrightarrow[]{a}q_2$ from $q_1$ to $q_2$  where $q_1\in Q_1, q_2\in Q_2$.
We first prove by contradiction that $a=a_0$ and $q_1=p_0$: It it is not the case, then the transition $q_1\xrightarrow[]{a}q_2$ belongs to a cycle labelled by $a$, which is not possible
as such a cycle would contain a transition from $Q_2$ to $Q_1$. Hence $p_0\xrightarrow[]{a_0}q_2$ is the only transition from $Q_1$ to $Q_2$ in this case. We obtain an upper bound
of the number of automata in this case by (1) choosing the $r$ states in $Q_1$, (2) choosing
the actions of all letters but $a_0$ on both $Q_1$ and $Q_2$, (3) choosing the dangling state
$p_0$ in $Q_1$ and its image by $a_0$ in $Q_2$, and (4) choosing the action of $a_0$ on
$Q_1\setminus\{p_0\}$ and on $Q_2$. As in Lemma~\ref{lem:sc_group_automata}, this yields an upper bound, not an exact counting, since some automata are counted several times this way. Therefore, the upper bound we obtain in this case, for fixed $r\in\{1,\ldots,n-1\}$ is:
\begin{equation}\label{eq:case1}
\binom{n}{r}r!^{k-1}(n-r)!^{k-1} r(n-r) (r-1)!(n-r)! = \binom{n}{r} r!^k(n-r)!^k(n-r).
\end{equation}

\noindent $\triangleright$ If there is no transition from $Q_1$ to $Q_2$, the automaton is really
split in two (or more) components, even as an undirected graph. The dangling state is either in $Q_1$ or $Q_2$, we can assume it is in $Q_1$ by symmetry. The restriction to $Q_1$ of the automaton is an almost-group automaton, and the restriction to $Q_2$ is a group automaton. We therefore have the following upper bound for such automata for fixed $r=|Q_1|$, using Lemma~\ref{lem:ag automata}:
\begin{equation}\label{eq:case2}
\binom{n}{r}|\G_{r,k}|(n-r)!^k = \binom{n}{r}(r-1)r!^k(n-r)!^k.
\end{equation}
To conclude the proof, we just have to sum everything up for $1 \leq r \leq n-1$:
\begin{equation}
\sum_{r=1}^{n-1}{n \choose r}r!^k(n-r)!^{k}(r-1 + n-r) = (n-1)\sum_{r=1}^{n-1} \binom{n}{r}r!^k(n-r)!^{k}.
\end{equation}
This concludes the proof as the sum is equal to twice the value of $Z_{n,k}$ of the proof of Lemma~\ref{lem:sc_group_automata}, possibly with a negligible extra central value.\qed
\end{proof}

\end{document}